\documentclass[aps,pra,reprint,superscriptaddress]{revtex4-1}
\usepackage{blindtext}
\usepackage{amssymb}
\usepackage{amsmath}
\usepackage{hhline}
\usepackage{physics}
\usepackage{amsthm}
\usepackage{tikz}
\usepackage{graphicx}
\usepackage{siunitx}

\newtheorem*{claim*}{Claim}

\newtheorem{prop}{Proposition}\def\PRO{\begin{prop}}\def\ORP{\end{prop}}
\newtheorem{coro}{Corollary}\def\COR{\begin{coro}}\def\ROC{\end{coro}}
\newtheorem{theo}{Theorem}\def\TH{\begin{theo}}\def\HT{\end{theo}}
\def\TH{\begin{theo}}\def\HT{\end{theo}}
\newtheorem{defi}[prop]{Definition}\def\DE{\begin{defi}}\def\ED{\end{defi}}

\newtheorem{lemme}[prop]{Lemma}\def\LE{\begin{lemme}}\def\EL{\end{lemme}}
\usepackage{color}
\usepackage{graphicx,float}
\usepackage{epsfig}
\usepackage{amsmath,amssymb}
\usepackage{bm}
\usepackage[latin1]{inputenc}
\usepackage{hyphenat}
\usepackage[bookmarks=false]{hyperref}
\usepackage{color}
\usepackage[capitalise]{cleveref}
\usepackage{epstopdf}
\usepackage{braket}
\usepackage{verbatim}
\usepackage{comment}
\usepackage{multirow}
\usepackage{float}
\usepackage{enumitem} 
\usepackage[caption=false]{subfig}

\def\ket#1{\left| #1 \right\rangle}
\def\bra#1{\left\langle #1 \right|}

\newcommand{\beq}{\begin{equation}}
\newcommand{\eeq}{\end{equation}}

\definecolor{pink}{RGB}{255,0,255}

\definecolor{ss_color}{rgb}{0,0,1}

\definecolor{darkorange}{RGB}{255,120,0} 

\definecolor{red}{rgb}{1,0,0}

\begin{document}

\title{Zero-error attack against coherent-one-way quantum key distribution}
\author{R\'obert Tr\'enyi}
\affiliation{Escuela de Ingenier\'ia de Telecomunicaci\'on, Department of Signal Theory and Communications, University of Vigo, Vigo E-36310, Spain}
\author{Marcos Curty}
\affiliation{Escuela de Ingenier\'ia de Telecomunicaci\'on, Department of Signal Theory and Communications, University of Vigo, Vigo E-36310, Spain}

\begin{abstract}
Coherent-one-way (COW) quantum key distribution (QKD) held the promise of distributing secret keys over long distances with a simple experimental setup. Indeed, this scheme is currently used in commercial applications. Surprisingly, however, it has been recently shown that its secret key rate scales at most quadratically with the system's transmittance and, thus, it is not appropriate for long distance QKD transmission. Such pessimistic result was derived by employing a so-called zero-error attack, in which the eavesdropper does not introduce any error, but still the legitimate users of the system cannot distill a secure key. Here, we present a zero-error attack against COW-QKD that is essentially optimal, in the sense that no other attack can restrict further its maximum achievable distance in the absence of errors. This translates into an upper bound on its secret key rate that is more than an order of magnitude lower than previously known upper bounds. 
\end{abstract}

\maketitle

\section{Introduction}

Quantum key distribution (QKD)~\cite{qkd1,qkd2} is probably the most mature quantum technology today, with QKD networks being deployed worldwide~\cite{net1,net2,net3,net4}. These networks permit pairs of distant users (say Alice and Bob) to generate information-theoretic secure cryptographic keys, which can be used to achieve perfectly secure communications via the one-time-pad cryptosystem~\cite{vernam}. 

Until quantum repeaters~\cite{rep1,rep2,rep3,rep4} are experimentally available, QKD networks typically rely on a trusted-node architecture. This is so because channel loss poses strong limitations on the secret key rate that can be achieved with point-to-point QKD. Indeed, this key rate scales at best linearly with the system's transmittance $\eta$~\cite{TGW,PLOB}, which in fibre-based links decreases exponentially with the distance. Recent years have witnessed a tremendous effort to develop QKD protocols able to deliver such best possible key rate scaling with practical signals. 

The most popular solution to achieve a key rate of order $O(\eta)$ with laser sources is undoubtedly decoy-state QKD~\cite{decoy2,decoy3,decoy1}, where Alice randomly varies the intensity of the optical pulses she sends to Bob. This scheme has been recently used to distribute secret keys over a world record distance of 421 km with optical fibre~\cite{record_decoy}. Alternative approaches include QKD with strong reference pulses~\cite{strong1,strong2,strong3}, and distributed-phase-reference (DPR) QKD. While it has been proven that the  key rate of the former scales also linearly with $\eta$, for DPR-QKD this has been demonstrated only against restricted types of attacks~\cite{cow4}. 

The key advantage of DPR-QKD when compared to other solutions is its simple experimental implementation. There are two main protocols: differential-phase-shift (DPS) QKD~\cite{dps1,dps2,dps3}, and coherent-one-way (COW) QKD~\cite{cow4,cow1,cow2,cow3}. Long distance implementations of both schemes have been reported recently, over 200 km~\cite{dps3} and 300 km~\cite{cow4}, respectively. Remarkably, there are even commercial systems based on COW-QKD~\cite{IdQ}. However, despite this significant progress, the security of DPR-QKD has yet to be fully established. For instance, it has been shown that DPS-QKD can offer a key rate of order $O(\eta^{3/2})$ in the long distance regime given that the error rate is kept sufficiently small~\cite{dps5,dps6}.  Also, it turns out that a key rate of order almost $O(\eta)$ is possible when the receiver checks the coherence between randomly chosen incoming signals~\cite{dps4,dps7,dps8}. Surprisingly, however, for COW-QKD it has been recently proven that its key rate scales at most quadratically with $\eta$~\cite{cow_attack}, which renders this scheme inappropriate for long distance QKD transmission. Indeed, this result matches the scaling of the lower security bound derived in~\cite{low_cow}. 

This pessimistic result was derived by employing a special type of sequential attack~\cite{seq1,seq2,seq3,gain}, which is called a zero-error attack~\cite{cow_attack,upp_cow2} because it does not introduce errors. In a sequential attack, the eavesdropper (Eve) first measures out all the signals sent by Alice one-by-one, and then she sends Bob signals which depend on all her measurement results. Importantly, when the measurement statistics observed by Alice and Bob are compatible with a sequential attack, they cannot distill a secret key. This is so because sequential attacks transform the quantum channel into an entanglement breaking channel and, thus, they do not allow Alice and Bob to establish quantum correlations between them~\cite{condition}. 

Here, we present a zero-error attack against COW-QKD which is essentially optimal, in the sense that no other attack can restrict further its maximum achievable distance in the absence of errors. That is, the attack introduced below provides the highest possible value of the system's transmittance underneath which COW-QKD is insecure. For example, it can be shown that if Alice and Bob use state-of-the-art devices and the intensity of Alice's signals is similar to that employed in decoy-state QKD, the maximum achievable distance of COW-QKD is only about $22$~km. Our findings translate into an upper bound on the secret key rate of COW-QKD that is more than an order of magnitude lower than previously known upper bounds~\cite{cow_attack,upp_cow2,upp_cow1,filteringattack}. 

The paper is structured as follows. In Sec.~\ref{Sec:COW} we introduce COW-QKD. Then, in Sec.~\ref{Sec:Attack} we present our zero-error attack against this scheme in detail. This attack uses an unambiguous state discrimination (USD) measurement~\cite{chefles_usd1,chefles_usd2,eldar1}, which we introduce in Sec.~\ref{meas}. Next, in Sec.~\ref{Gth}, we obtain the maximum value of the gain at Bob's side that is achievable with the zero-error attack. This quantity is directly related to the maximum value of the system's transmittance below which COW-QKD is insecure. In Sec.~\ref{Sec:Evaluation}, we evaluate the implications of these results on the distance and secret key rate that can be achieved with COW-QKD, and we compare our results with other zero-error attacks. Finally, we discuss possible countermeasures against zero-error attacks in Sec.~\ref{Sec:Discussion}, and conclude the paper with a summary in Sec.~\ref{Sec:Conclusion}. The paper also includes two appendixes with additional calculations. 

\section{Coherent-one-way quantum key distribution}\label{Sec:COW} 

The layout of COW-QKD~\cite{cow4,cow1,cow2,cow3} is depicted in Fig.~\ref{COW}. 
\begin{figure}
\centering{\includegraphics*[scale=0.42]{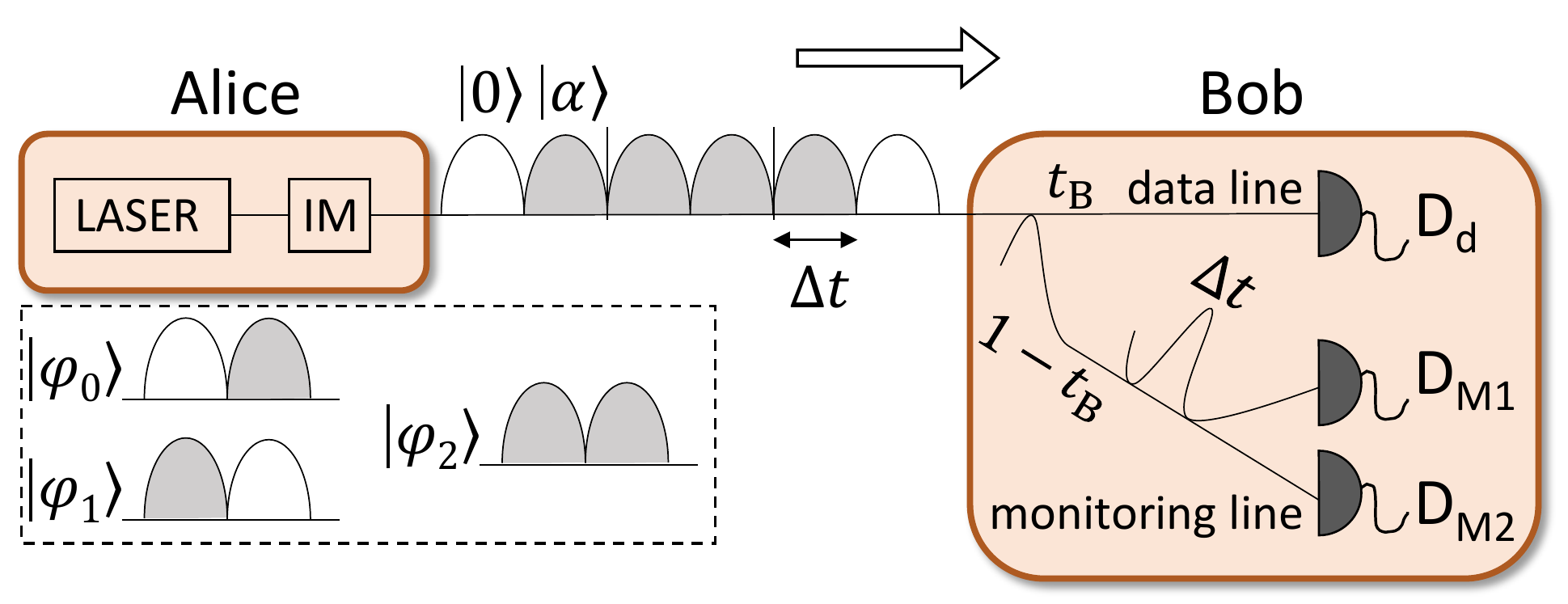}}
\caption{Schematics of COW-QKD. Alice sends Bob a random sequence of signals $\ket{\varphi_i}$, with $i=0,\ldots,2$. Bob employs a beamsplitter with transmittance $t_{\rm B}$ to passively distribute the incoming signals between the data and the monitoring lines. The detection events at the data line constitute the raw key, while the monitoring line is used to detect eavesdropping. This is done by means of a Mach-Zehnder interferometer that measures the coherence between adjacent pulses. In the figure, a grey (white) oval represents a coherent (vacuum) state $\ket{\alpha}$ ($\ket{0}$), IM is an intensity modulator, $\Delta{}t$ is the time delay between subsequent pulses, and D$_{\rm d}$, D$_{\rm M1}$ and D$_{\rm M2}$ are single-photon detectors.}
\label{COW}
\end{figure}
Alice sends Bob a random sequence of signals $\ket{\varphi_0}=\ket{0}\ket{\alpha}$, $\ket{\varphi_1}=\ket{\alpha}\ket{0}$ and $\ket{\varphi_2}=\ket{\alpha}\ket{\alpha}$, where $\ket{0}$ ($\ket{\alpha}$) represents a vacuum (coherent) state. She prepares these signals with a priori probabilities $p_{\ket{\varphi_0}}=p_{\ket{\varphi_1}}=(1-f)/2$, and $p_{\ket{\varphi_2}}=1-p_{\ket{\varphi_0}}-p_{\ket{\varphi_1}}=f$, with $f\in (0,1)$. A signal $\ket{\varphi_0}$ ($\ket{\varphi_1}$) encodes a bit value $0$ ($1$), while $\ket{\varphi_2}$ represents a decoy signal. 

At the receiving side, Bob uses a beamsplitter of transmittance $t_{\rm B}$ to passively distribute the incoming signals between the data and the monitoring lines. The data line allows him to distinguish the states $\ket{\varphi_0}$ and $\ket{\varphi_1}$ by measuring the time instance in which detector D$_{\rm d}$ provides a detection ``click'' (see Fig.~\ref{COW}). Precisely, Bob assigns a bit value $0$ ($1$) to a ``click'' in the first (second) time slot within a signal, while he assigns a random bit value to a double ``click''. These bits constitute his raw key. 

Next, Bob announces over an authenticated classical channel which signals produced a ``click'' in D$_{\rm d}$, but he does not reveal the particular time slots in which the detection ``clicks'' actually occurred. Also, Alice declares if Bob's observed ``clicks'' correspond to bit states ({\it i.e.}, to the states $\ket{\varphi_0}$ and $\ket{\varphi_1}$) or to decoy signals. The bits associated to the bit states constitute the sifted key.

The monitoring line, on the other hand, is used by Bob to measure the coherence between adjacent coherent states $\ket{\alpha}$ to detect for eavesdropping. This is achieved with a Mach-Zehnder interferometer followed by two detectors, D$_{\rm M1}$ and D$_{\rm M2}$. This interferometer is such that two adjacent states $\ket{\alpha}$ cannot produce a ``click'' in say detector D$_{\rm M2}$.

Errors in the data line are characterized by means of the quantum bit error rate (QBER), while errors in the monitoring line are characterized with the visibilities
\begin{equation}
{V_s} = \frac{{p_{\rm click}({{\rm{D_{M1}}}|s}) - p_{\rm click}({{\rm{D_{M2}}}|s})}}{{p_{\rm click}({{\rm{D_{M1}}}|s}) + p_{\rm click}({{\rm{D_{M2}}}|s})}}, \label{vis_gen}
\end{equation}
with $s\in{\mathcal S}=\{``d$''$, ``01$''$, ``0d$''$, ``d1$''$, ``dd$"$\}$. In Eq.~(\ref{vis_gen}), $p_{\rm click}({\rm D}_{{\rm M}i}|s)$ represents the conditional probability that detector ${\rm D}_{{\rm M}i}$ ``clicks'' given that the two adjacent coherent states $\ket{\alpha}$ are situated within a sequence $s$ sent by Alice. For example, a sequence $s=``01$'' represents a bit 1 signal followed by a bit 0 signal ({\it i.e.}, $\ket{\varphi_0}\ket{\varphi_1}$), and the other sequences are defined similarly. 

\section{Zero-error attack}\label{Sec:Attack}

The zero-error attack that we introduce below exploits two special properties of Alice's signals. First, they are linearly independent, and, second, they contain the vacuum state, which naturally breaks the coherence between adjacent pulses. Indeed, thanks to these two properties together, it is possible for Eve to perform a sequential attack, based on an USD measurement~\cite{chefles_usd1,chefles_usd2,eldar1}, which does not introduce any error in the data line nor in the monitoring line and still prevents Alice and Bob from distilling a secure key. 

Various zero-error attacks against COW-QKD have been proposed recently in~\cite{upp_cow2,cow_attack}. For instance, as already mentioned, the work in~\cite{cow_attack} uses a particular zero-error attack to show that the secret key rate of COW-QKD scales at most quadratically with the overall system's transmittance $\eta$. 

A crucial parameter in a zero-error attack is the maximum value of the gain at Bob's data line for which the attack is actually possible. We shall denote this parameter by $G_{\rm zero}$. Here, the gain is defined as the probability that Bob observes a detection ``click'' in his data line per signal state sent by Alice. Precisely, in a zero-error attack Eve first measures out each signal sent by Alice with a USD measurement, and then she essentially resends Bob those signals for which she obtained a successful measurement result. If the measurement outcome is inconclusive, Eve sends Bob vacuum states. This strategy naturally imposes a maximum value of the gain at Bob's data line that is achievable with a zero-error attack due to the inconclusive outcomes provided by Eve's measurement. Note that here we are considering the conservative {\it untrusted device} scenario. Therefore, a resent vacuum signal can never cause a detection ``click'' at Bob's side, as this scenario assumes that the detection efficiency and dark count rate of Bob's detectors can be controlled by Eve. We shall denote the maximum value of $G_{\rm zero}$ among all possible zero-error attacks by $G_{\rm zero}^{\rm max}$.

Next, we introduce a simple zero-error attack against COW-QKD for which $G_{\rm zero}$ can be made arbitrarily close to the optimal value $G_{\rm zero}^{\rm max}$. As we will show in detail in Sec.~\ref{Sec:UpperBound}, this translates into an upper bound on the secret key rate of COW-QKD that is more than an order of magnitude lower than previous results~\cite{upp_cow2,cow_attack}.

The key idea of the attack is rather elementary; it goes as follows. First, Eve measures the signals sent by Alice one by one with an optimal USD measurement that maximizes her probability to obtain a conclusive result. This measurement is described in detail in Sec.~\ref{meas}. Let $k$ denote the number of consecutive conclusive measurement results obtained by Eve with her measurement. If $k=0$, {\it i.e.} when the signal measured by Eve results in an inconclusive output, Eve directly replaces this signal with a vacuum signal $\ket{\varphi_{\rm vac}}=\ket{0}\ket{0}$ and sends $\ket{\varphi_{\rm vac}}$ to Bob. Moreover, she restarts the counting of the number of consecutive conclusive measurement results again from the next signal on. We note that this latter step is common to all the other cases that we discuss below and we omit it in their description for simplicity. 

If $k=1$, {\it i.e.} when the first signal measured by Eve results in a conclusive output but the following signal results in an inconclusive output, then Eve replaces these two signals with two vacuum signals $\ket{\varphi_{\rm vac}}$, she sends these latter states to Bob. 

If $2\leq k < M_{\rm max}$, {\it i.e.} when Eve obtained $k$ consecutive conclusive measurement results and the following signal is an inconclusive output, then she replaces this latter signal with a vacuum state $\ket{\varphi_{\rm vac}}$. Here, $M_{\rm max}$ is a pre-fixed value that Eve can select arbitrarily large. Moreover, Eve looks for the longest sub-block, within the block of $k$ correctly identified signals, which is surrounded by vacuum {\it pulses}. Then, she replaces the signals that do not belong to this sub-block with vacuum signals $\ket{\varphi_{\rm vac}}$. Also, she replaces the coherent states $\ket{\alpha}$ within the signals that belong to the successful sub-block, with coherent states $\ket{\beta}$ satisfying $\beta\gg\alpha$. We note that by selecting $\beta$ large enough, Eve can guarantee that each signal within the successful sub-block will produce a detection ``click'' at Bob's side with nearly unit probability. Alternatively, Eve could also replace the coherent states $\ket{\alpha}$ with signals that do not contain the vacuum state. However, here we prefer to use coherent states $\ket{\beta}$ just for simplicity. If a block of $k$ signals correctly identified by Eve does not contained a sub-block surrounded by vacuum pulses, then Eve replaces all the signals within the block with vacuum signals. That is, she sends Bob $k+1$ vacuum signals $\ket{\varphi_{\rm vac}}$. 

Finally, if $k=M_{\rm max}$, {\it i.e.} when Eve obtained $M_{\rm max}$ consecutive conclusive measurement results, then she directly replaces the following signal ({\it i.e.}, the signal that is located in the position $M_{\rm max}+1$) with a vacuum signal $\ket{\varphi_{\rm vac}}$. Moreover, she post-processes the block of $M_{\rm max}$ consecutive conclusive measurement results like in the previous case. That is, she looks for the longest sub-block that is surrounded by vacuum pulses within the block, and all signals that do not belong to such sub-block are replaced with vacuum signals $\ket{\varphi_{\rm vac}}$. Also, she replaces the coherent states $\ket{\alpha}$ within the successful sub-block with coherent states $\ket{\beta}$, and sends the new block of $M_{\rm max}+1$ signals to Bob. 

It is clear that the attack above is indeed a zero-error attack. Note that the resulting QBER is zero because Eve never misidentifies a signal sent by Alice. Moreover, the attack also preserves the coherence between those original adjacent coherent states $\ket{\alpha}$ prepared by Alice and which are resent to Bob by Eve. This guarantees that Bob will obtain $V_s=1$ for $\forall s\in{\mathcal S}$. 

What is more, as already mentioned, the gain $G_{\rm zero}$ achieved with the attack above can be made arbitrarily close to $G_{\rm zero}^{\rm max}$ by simply selecting the parameter $M_{\rm max}$ large enough. This is so because, in such regime, all signals which are correctly identified by the optimal USD measurement, and which do not reduce the visibility at Bob's side, actually produce a detection ``click'' in his data line. Here, it is important to note that, since Alice prepares her signals at random and independently of each other, measuring these signals individually (as it is done in the attack above) is not disadvantageous for Eve when compared to a possible USD strategy based on joint measurements. This is so because, due to the independence of the signals, to unambiguously discriminate a group of signals (from other group of signals) each individual signal has to be discriminated unambiguously. Moreover, in practice, even relatively small values of $M_{\rm max}$ ({\it e.g.} $M_{\rm max}=10$, which is the value considered in the simulations shown in Sec.~\ref{Sec:Evaluation}) are sufficient to basically achieve $G_{\rm zero}^{\rm max}$. This is due to the fact that, when $\alpha$ is small (as is typically the case in experimental implementations of COW-QKD), the probability that Eve obtains more than $M_{\rm max}$ consecutive conclusive measurement results is essentially negligible even for moderate values of $M_{\rm max}$. 

\section{USD measurement}\label{meas}

In this section, we now describe Eve's optimal USD measurement, and provide an analytical expression for the maximum probability to obtain a conclusive result, which we shall denote by $p_{\rm c}$. 

Precisely, Eve's measurement contains four measurement operators $E_j\geq0$ satisfying $\sum_{j=0}^3 E_j=\openone$, with $\openone$ denoting the identity operator. A result associated to $E_j$, with $j=0,1,2$, unambiguously identifies the state $\ket{\varphi_j}$ sent by Alice, while a result associated to $E_3$ corresponds to an inconclusive result. 

Let $p_{j|i}=\bra{\varphi_i}E_j\ket{\varphi_i}$ denote the conditional probability that Eve obtains the result $j$ given that Alice sent the state $\ket{\varphi_i}$. Since we are considering a USD measurement, we have that $p_{j|i}=0$ $\forall i\neq{}j$ with $i,j=0,1,2$. Moreover, since Alice's signals $\ket{\varphi_0}$ and $\ket{\varphi_1}$ are sent with the same a priori probability, it can be shown that the optimal USD measurement satisfies $p_{0|0}=p_{1|1}\equiv{}q^{\rm s}_{\rm s}$. These probabilities are illustrated in Table~\ref{table1}.
\begin{table}
  \centering
  \begin{tabular}{lcccc}
    \quad & \multicolumn{4}{c}{Eve's POVM elements}\\
    \hline\hline
  Alice's signal & \quad $E_0$ \quad & \quad $E_1$ \quad & \quad $E_2$ \quad & \ \quad $E_3$  \ \quad \\
  \hline
\quad\quad$\ket{\varphi_0}$   & \quad $q^{\rm s}_{\rm s}$ \quad & \quad 0 \quad  & \quad 0 \quad  & $q^{\rm s}_{\rm inc}$   \quad  \\
\quad\quad$\ket{\varphi_1}$   & \quad 0 \quad & \quad $q^{\rm s}_{\rm s}$ \quad  & \quad 0 \quad  & $q^{\rm s}_{\rm inc}$ \quad  \\
\quad\quad$\ket{\varphi_2}$   & \quad 0 \quad & \quad 0 \quad  & \quad $q^{\rm d}_{\rm s}$ \quad  & $q^{\rm d}_{\rm inc}$  \quad  \\ 
 \hline\hline
\end{tabular}
\caption{Conditional probabilities associated to Eve's USD measurement. For convenience, we use the labels $q^{\rm s}_{\rm s}$, $q^{\rm d}_{\rm s}$, $q^{\rm s}_{\rm inc}$ and $q^{\rm d}_{\rm inc}$ to denote the different conditional probabilities. They satisfy $q^{\rm s}_{\rm s}+q^{\rm s}_{\rm inc}=q^{\rm d}_{\rm s}+q^{\rm d}_{\rm inc}=1$.
  }\label{table1}
\end{table}

We find, therefore, that $p_{\rm c}$ satisfies
\begin{equation}\label{qs_qf_qinc}
p_{\rm c}=\sum_{i=0}^2 p_{\ket{\varphi_i}}p_{i|i}=(1-f)q^{\rm s}_{\rm s}+fq^{\rm d}_{\rm s},
\end{equation}
while the probability to obtain an inconclusive result is simply given by $p_{\rm inc}=1-p_{\rm c}$. In Eq.~(\ref{qs_qf_qinc}), we use the values for the probabilities $p_{\ket{\varphi_i}}$ that are given in Sec.~\ref{Sec:COW}, and we use the notation introduced in Table~\ref{table1} for the probabilities $p_{i|i}$. 

Now, to obtain the maximum possible value of $p_{\rm c}$, we follow the techniques introduced in~\cite{sugimoto}. The result is given by the following Claim. 
\begin{claim*}
The maximum probability $p_{\rm c}$ to obtain a conclusive result when unambiguously discriminating the signals $\ket{\varphi_0}=\ket{0}\ket{\alpha}$, $\ket{\varphi_1}=\ket{\alpha}\ket{0}$ and $\ket{\varphi_2}=\ket{\alpha}\ket{\alpha}$, which are sent with a priori probabilities $p_{\ket{\varphi_0}}=p_{\ket{\varphi_1}}=(1-f)/2$ and $p_{\ket{\varphi_2}}=f$, with $f\in (0,1)$, is given by Eq.~(\ref{qs_qf_qinc}) with the conditional probabilities $q^{\rm s}_{\rm s}$ and $q^{\rm d}_{\rm s}$ satisfying: 
\\

If $\sqrt{\gamma}\leq e^{-|\alpha|^2/2}$, with $\gamma=f/[2(1-f)]$, then $q^{\rm s}_{\rm s}=1-e^{-|\alpha|^2}$ and $q^{\rm d}_{\rm s}=0$;
\\ 

If $\sqrt{\gamma}> e^{-|\alpha|^2/2}$ and $\cosh\left(|\alpha|^2/2\right)\geq \sqrt{\gamma}$ then $q^{\rm s}_{\rm s}=1+e^{-|\alpha|^2}-e^{-|\alpha|^2/2}\sqrt{2f/(1-f)}$ and $q^{\rm d}_{\rm s}=1-\sqrt{\gamma^{-1}}e^{-|\alpha|^2/2}$;

\quad
\\
Finally, if $\sqrt{\gamma}> e^{-|\alpha|^2/2}$ and $\cosh\left(|\alpha|^2/2\right)< \sqrt{\gamma}$ then $q^{\rm s}_{\rm s}=0$ and $q^{\rm d}_{\rm s}=\tanh\left(|\alpha|^2/2\right)$.
\end{claim*}
\begin{proof}[Proof.]
  The proof of the Claim is provided in Appendix~\ref{proof}. 
\end{proof}

We note that in most experiments~\cite{cow4,cow3}, we have typically that $f\leq 0.15$ and $|\alpha|^2\leq 0.5$. This means that the first condition in the Claim is usually satisfied.

\section{Gain $G_{\rm zero}$}\label{Gth}

In this section we calculate $G_{\rm zero}$ for the zero-error attack introduced in Sec.~\ref{Sec:Attack}. As already explained in that section, this parameter represents the maximum value of the gain at Bob's data line that is achievable with Eve's zero-error attack. That is, whenever the observed gain of an experimental implementation of COW-QKD is below $G_{\rm zero}$, the protocol is insecure~\cite{condition}.

Our starting point is the definition of the gain at Bob's data line, which can be expressed as $G=N_{\rm clicks}/N$. Here, $N_{\rm clicks}$ is the total number of clicks observed by Bob in his data line, and $N$ is the total number of signals sent by Alice. In the asymptotic limit where $N$ tends to infinity, $N_{\rm clicks}$ can be written as $N_{\rm clicks}=(N/N^{\rm av})N_{\rm clicks}^{\rm av}$, where $N^{\rm av}$ denotes the average length of the blocks of signals that Eve sends to Bob, and $N_{\rm clicks}^{\rm av}$ represents the average number of ``clicks'' observed by Bob in his data line due to these blocks of signals~\cite{cow_attack,gain}. With this notation, one can rewrite the gain as
\begin{equation}\label{gain1}
G=\frac{N_{\rm clicks}^{\rm av}}{N^{\rm av}}.
\end{equation}
Next, we calculate the parameters $N^{\rm av}$ and $N_{\rm clicks}^{\rm av}$. 

Let us start with $N^{\rm av}$. The a priori probability that Eve sends Bob a block containing $k+1$ signals, with $k=2,\ldots,M_{\rm max}$, in which the first $k$ signals provided Eve with a conclusive measurement result and the last one is a vacuum signal $\ket{\varphi_{\rm vac}}$, is given by
\begin{equation}\label{psk}
p_{\rm s}(k) = \left\{ \begin{array}{ll}
p_{\rm c}^k(1-p_{\rm c}) & \textrm{if $2\leq k<M_{\rm max}$,}\\
p_{\rm c}^{M_{\rm max}} & \textrm{if $k=M_{\rm max}$,}\\
0 & \textrm{otherwise,}
  \end{array} \right.
\end{equation}
with the probability $p_{\rm c}$ having the form of Eq.~(\ref{qs_qf_qinc}).

Similarly, the probability that Eve sends Bob a block with $k+1$ vacuum signals $\ket{\varphi_{\rm vac}}$, with $k=0,1$, is given by
\begin{equation}\label{pvk}
p_{\rm v}(k) = \left\{ \begin{array}{ll}
p_{\rm c}^k(1-p_{\rm c}) & \textrm{if $0\leq{}k\leq{}1$,}\\
0 & \textrm{otherwise.}
  \end{array} \right.
\end{equation}

We find, therefore, that $N^{\rm av}$ can be written as
\begin{equation}\label{NE}
N^{\rm av}=\sum_{k=0}^{1} p_{\rm v}(k) (k+1)+\sum_{k=2}^{M_{\rm max}} p_{\rm s}(k) (k+1).
\end{equation}
By substituting Eqs.~(\ref{psk})-(\ref{pvk}) into Eq.~(\ref{NE}) we obtain that 
\begin{equation}\label{NE_final}
N^{\rm av}=\frac{1-p_{\rm c}^{M_{\rm max}+1}}{1-p_{\rm c}}.
\end{equation}

Next, we calculate $N_{\rm clicks}^{\rm av}$. For this, we need to consider only those blocks of signals that Eve sends to Bob containing at least one non-vacuum signal ({\it i.e.}, a signal $\ket{\varphi_i}$ with $i=0,1,2$). This is so because, as already mentioned before, in the untrusted device scenario the vacuum signals $\ket{\varphi_{\rm vac}}$ cannot produce a ``click'' at Bob's side. This means, in particular, that $N_{\rm clicks}^{\rm av}$ can be written as
\begin{equation}\label{NEc}
N_{\rm clicks}^{\rm av}=\sum_{k=2}^{M_{\rm max}} p_{\rm s}(k) p_{\rm click}(k),
\end{equation}
where $p_{\rm s}(k)$ is given by Eq.~(\ref{psk}), and $p_{\rm click}(k)$ denotes the average number of ``clicks'' observed by Bob in his data line when Eve sends him a block containing $k+1$ signals, with $k=2,\ldots,M_{\rm max}$. To calculate this latter quantity, we need to take into account the post-processing step that Eve applies to the blocks of signals before she sends them to Bob. This is what we do in the last part of this section. 

But before we calculate $p_{\rm click}(k)$, we note that $G_{\rm zero}$ corresponds to the maximum value of $G$, which happens when we select the maximum probability $p_{\rm c}$ given by Eq.~(\ref{qs_qf_qinc}). This corresponds to using the parameters $q^{\rm s}_{\rm s}$ and $q^{\rm d}_{\rm s}$ given by the Claim in Sec.~\ref{meas}. From Eqs.~(\ref{gain1})-(\ref{psk})-(\ref{NE_final})-(\ref{NEc}), we obtain, therefore, that 
\begin{eqnarray}
G_{\rm zero}&=&\frac{1-p_{\rm c}}{1-p_{\rm c}^{M_{\rm max}+1}}\Bigg[\sum_{k=2}^{M_{\rm max}-1} p_{\rm c}^k(1-p_{\rm c}) p_{\rm click}(k)\nonumber \\
&+&p_{\rm c}^{M_{\rm max}}p_{\rm click}(M_{\rm max})\Bigg].
\end{eqnarray}
As already mentioned earlier, $G_{\rm zero}$ increases when we increase $M_{\rm max}$ and converges very quickly to $G_{\rm zero}^{\rm max}$. Indeed, above a certain relatively small value of $M_{\rm max}$, the improvement is already essentially negligible.

\subsection{Probabilities $p_{\rm click}(k)$}

To calculate $p_{\rm click}(k)$, we shall consider three cases, depending on the measurement result associated to the signal located in the first position of the block. This is illustrated in Fig.~\ref{fig:firstSignal}.
\begin{figure}
\centering 
\centerline{\includegraphics*[scale=0.5]{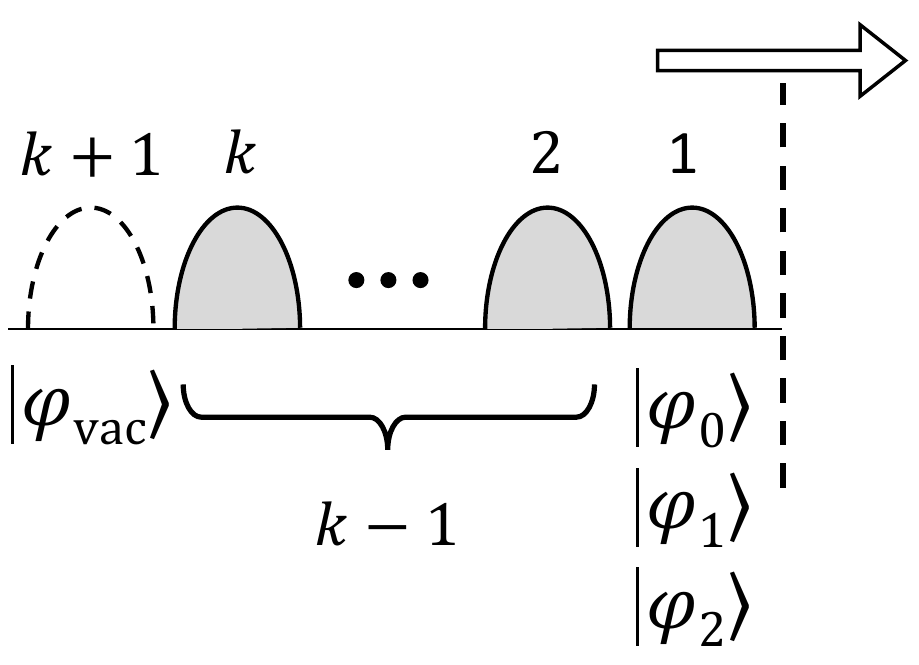}}
\caption{Schematic representation of a block of signals in which Eve obtained $k$ consecutive conclusive measurement results (indicated in the figure with grey oval signals), with $k=2,\ldots,M_{\rm max}$, and the signal at the position $k+1$ is replaced with a vacuum signal $\ket{\varphi_{\rm vac}}$. To calculate $p_{\rm click}(k)$, we consider three cases, depending on the result ($\ket{\varphi_0}$, $\ket{\varphi_1}$ or $\ket{\varphi_2}$) obtained by Eve in the first position of the block. The arrow indicates the direction of transmission towards Bob, and the dashed vertical line represents the point where Eve started counting the number of consecutive conclusive measurement results. }
\label{fig:firstSignal}
\end{figure}

In the calculations below, we will use the conditional probabilities, $p(j|{\rm c})$ with $j=0,1,2$, that Eve obtains a result $\ket{\varphi_j}$ given that her USD measurement is conclusive. From Table~\ref{table1}, together with the a priori probabilities $p_{\ket{\varphi_j}}$ that Alice generates the signals $\ket{\varphi_i}$, it is straightforward to show that the probabilities $p(j|{\rm c})$ satisfy
\begin{eqnarray}\label{cond_prob}
p(0|{\rm c})&=&p(1|{\rm c})=\frac{(1-f)q^{\rm s}_{\rm s}}{2p_{\rm c}}, \ {\rm and}\ p(2|{\rm c})=\frac{f q^{\rm d}_{\rm s}}{p_{\rm c}}. \nonumber \\
\end{eqnarray}
Obviously, they fulfill $\sum_{j=0}^2 p(j|{\rm c})=1$.

Depending on the measurement result that Eve obtains for the first signal in a block, we can write $p_{\rm click}(k)$ as follows
\begin{eqnarray}\label{pclickFirstSignal}
p_{\rm click}(k)&=&\sum_{j=0}^2p(j|{\rm c})p_{\rm click}(k|j)\nonumber \\
&=&p(1|{\rm c})\sum_{j=0}^1p_{\rm click}(k|j)+\big[1-2p(1|{\rm c})\big]p_{\rm click}(k|2), \nonumber \\
\end{eqnarray}
where the conditional quantities $p_{\rm click}(k|j)$, with $j=0,1,2$, denote the average number of ``clicks'' observed by Bob when Eve sends him a block with $k+1$ signals in which she observed the signal $\ket{\varphi_j}$ in the first position of the block. Also, in the second equality within Eq.~(\ref{pclickFirstSignal}) we use the properties of the conditional probabilities $p(j|{\rm c})$.

As already mentioned earlier, to calculate $p_{\rm click}(k|j)$, which we do next, we shall consider, for simplicity, that the intensity $|\beta|^2$ of the coherent pulses resent by Eve is sufficiently large such that Bob obtains a detection ``click'' with basically unit probability. This is implicitly assumed in the calculations that follow.

\subsubsection{Average number of ``clicks'' $p_{\rm click}(k|1)$}

In this case, since the first optical pulse of a signal $\ket{\varphi_1}$ is a vacuum pulse, the longest sub-block situated between vacuum pulses (if there is any) must include this signal $\ket{\varphi_1}$. This is depicted in Fig.~\ref{fig:firstSignal1}. This figure includes as well the number of ``clicks'' that Bob obtains in his data line for each of the three sub-cases considered in that figure, which depend on the signal found in the $k$-th position of the block. 
\begin{figure}
\centering 
\centerline{\includegraphics*[scale=0.5]{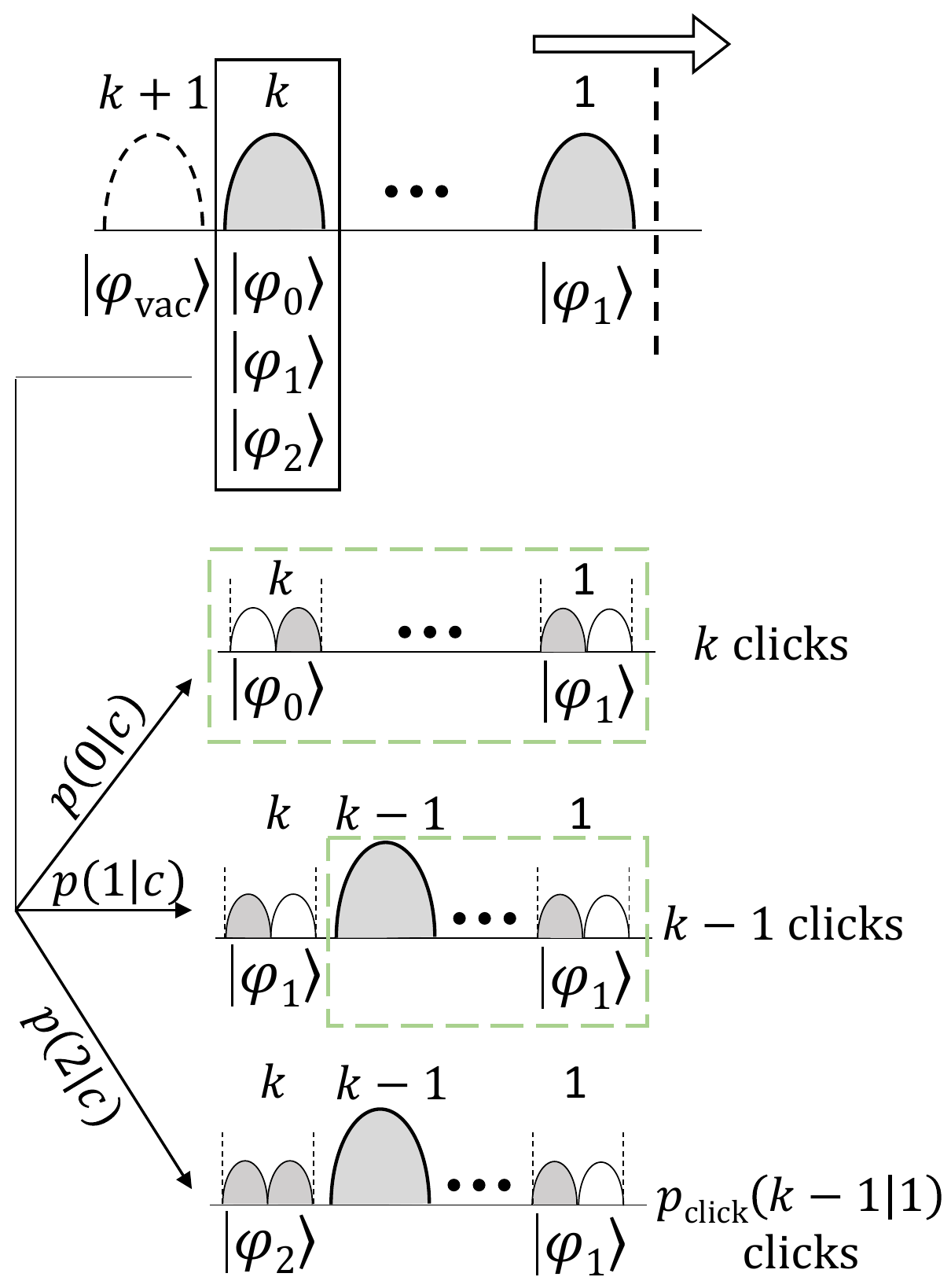}}
\caption{Illustration of the three sub-cases that we consider to evaluate $p_{\rm click}(k|1)$. With probability $p(0|{\rm c})$ the signal in the $k$-th position of the block is $\ket{\varphi_0}$. In this scenario, Eve resends Bob all the $k$ conclusive signals, as the block starts and ends with vacuum pulses. The number of ``clicks'' at Bob's data line is then $k$, because we assume that the intensity of Eve's signals is large enough such that they produce a detection ``click'' with basically unit probability and, moreover, double ``clicks'' within a signal are randomly assigned to single ``clicks'' by Bob. The other two sub-cases in which the signal in the $k$-th position of the block is $\ket{\varphi_1}$ or $\ket{\varphi_2}$ are  described in detail in the text. In the figure, large ovals represent signals, while small ovals represent optical pulses within a signal. Moreover, the dashed rectangles indicate the signals $\ket{\varphi_j}$, with $j=0,1,2$, that Eve resends to Bob.}
\label{fig:firstSignal1}
\end{figure}

Precisely, if the $k$-th signal is $\ket{\varphi_0}$, which happens with probability $p(0|{\rm c})$, then Eve resends Bob all the $k$ conclusive signals in the block because it starts and ends with vacuum pulses. This means that Bob will obtain $k$ detection ``clicks'', as double ``clicks'' are randomly assigned by him to single ``clicks''. On the other hand, if the $k$-th signal is $\ket{\varphi_1}$, which happens with probability $p(1|{\rm c})$, then Eve resends Bob the first $k-1$ conclusive signals in the block because such sub-block has vacuum pulses on its edges. At the same time, she replaces the $k$-th signal with $\ket{\varphi_{\rm vac}}$. This means that Bob will obtain $k-1$ detection ``clicks''. Finally, if the $k$-th signal is $\ket{\varphi_2}$, which happens with probability $p(2|{\rm c})$, she replaces that signal with $\ket{\varphi_{\rm vac}}$ because, obviously, the longest sub-block situated between vacuum pulses cannot contain such signal. Moreover, in this last scenario Bob will obtain $p_{\rm click}(k-1|1)$ detection ``clicks', as Eve now has a block with $\ket{\varphi_1}$ in its first position, followed by $(k-2)$-th signals $\ket{\varphi_j}$, with $j=0,1,2$. Note that when $k=2$ we have that $p_{\rm click}(k-1|1)=p_{\rm click}(1|1)=0$ because it is not possible to find a sub-block with at least one signal situated between vacuum pulses.

Putting all together, we obtain the following recursive relation for the expected number of ``clicks'' $p_{\rm click}(k|1)$ at Bob's side
\begin{eqnarray}\label{recursive_pc1}
p_{\rm click}(k|1)&=&p(0|{\rm c})k+p(1|{\rm c})\big[k-1\big]\nonumber \\
&+&p(2|{\rm c})p_{\rm click}(k-1|1)\nonumber \\
&=&p(1|{\rm c})\big[2k-1\big]\nonumber \\
&+&\big[1-2p(1|{\rm c})\big]p_{\rm click}(k-1|1),
\end{eqnarray} 
where in the second equality we use the properties of the conditional probabilities $p(j|{\rm c})$.

After some algebra, and taking into account that $p_{\rm click}(1|1)=0$, from Eq.~(\ref{recursive_pc1}) we obtain
\begin{align}\label{final_pc1}
p_{\rm click}(k|1)&=\frac{1}{2p(1|{\rm c})\big[2p(1|{\rm c})-1\big]}\nonumber\\&\times\bigg\{\big[1-2p(1|{\rm c})\big]^k\big[3p(1|{\rm c})-1\big]\nonumber\\
&+\big[2p(1|{\rm c})-1\big]\big[\big(2k+1\big )p(1|{\rm c})-1\big]\bigg\},
\end{align}
where $p(1|{\rm c})$ is given by Eq.~\eqref{cond_prob}.

\subsubsection{Average number of ``clicks'' $p_{\rm click}(k|0)$}

This case is very similar to the previous one, and we omit the details here for simplicity. The different sub-cases are illustrated in Fig.~\ref{fig:firstSignal1}. 
\begin{figure}
\centering 
\centerline{\includegraphics*[scale=0.5]{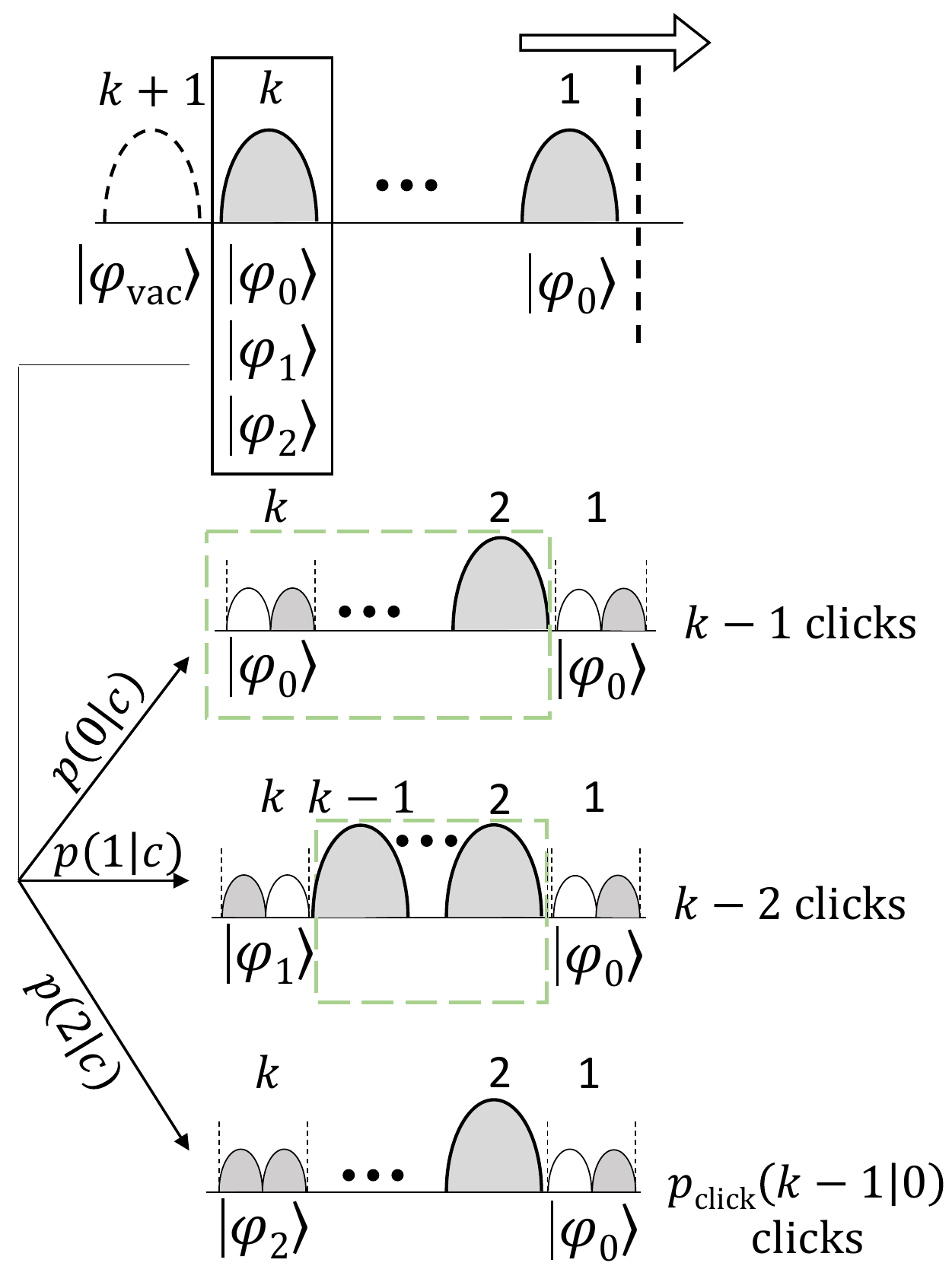}}
\caption{Illustration of the three sub-cases that we consider to evaluate $p_{\rm click}(k|0)$. With probability $p(0|{\rm c})$ the signal in the $k$-th position of the block is $\ket{\varphi_0}$. In this scenario, Eve resends Bob all the conclusive signals from position $2$ to position $k$, while the first signal is replaced with $\ket{\varphi_{\rm vac}}$. The number of ``clicks'' at Bob's side is then $k-1$. The other two sub-cases in which the signal in the $k$-th position of the block is $\ket{\varphi_1}$ or $\ket{\varphi_2}$ are described in detail in the text. For the meaning of the different elements see the caption of Fig.~\ref{fig:firstSignal1}.
}
\label{fig:firstSignal0}
\end{figure}
We find the following recursive relation for the expected number of ``clicks'' $p_{\rm click}(k|0)$ at Bob's side
\begin{eqnarray}\label{recursive_pc0}
p_{\rm click}(k|0)&=&p(0|{\rm c})\big[k-1\big]+p(1|{\rm c})\big[k-2\big]\nonumber \\
&+&p(2|{\rm c})p_{\rm click}(k-1|0)\nonumber \\
&=&p(1|{\rm c})\big[2k-3\big]\nonumber \\
&+&\big[1-2p(1|{\rm c})\big]p_{\rm click}(k-1|0).
\end{eqnarray} 
After some algebra, and taking into account that here $p_{\rm click}(1|0)=0$, from Eq.~(\ref{recursive_pc0}) we obtain
\begin{align}\label{final_pc0}
p_{\rm click}(k|0)&=\frac{1}{2p(1|{\rm c})\big[2p(1|{\rm c})-1\big]}\nonumber\\
&\times\bigg\{\big[1-2p(1|{\rm c})\big]^k\big[p(1|{\rm c})-1\big]\nonumber\\
&+\big[2p(1|{\rm c})-1\big]\big[\big(2k-1\big )p(1|{\rm c})-1\big]\bigg\},
\end{align}
where $p(1|{\rm c})$ is again given by Eq.~\eqref{cond_prob}.

\subsubsection{Average number of ``clicks'' $p_{\rm click}(k|2)$}

When the signal located in the first position of a block is $\ket{\varphi_2}$, Eve always replaces this signal with a vacuum signal $\ket{\varphi_{\rm vac}}$, as $\ket{\varphi_2}$ does not contain a vacuum pulse. Thus, the remaining sub-block has now $k-1$ conclusive results, each of which can be a signal $\ket{\varphi_j}$ with $j=0,1,2$. This means, therefore, that the average number of ``clicks'' $p_{\rm click}(k|2)$ coincides with that of a general block with $k-1$ consecutive conclusive measurement results. That is, we find that
\begin{equation}\label{final_pc2}
p_{\rm click}(k|2)=p_{\rm click}(k-1).
\end{equation}

We have now all the quantities required to evaluate $p_{\rm click}(k)$. Precisely, by combining Eqs.~\eqref{final_pc1}-\eqref{final_pc0}-\eqref{final_pc2} with Eq.~\eqref{pclickFirstSignal}, we obtain the following recursive relation for $p_{\rm click}(k)$,
\begin{align}\label{final_pClick}
p_{\rm click}(k)&=2kp(1|{\rm c})+\big[1-2p(1|{\rm c})\big]^k-1\nonumber\\&+\big[1-2p(1|{\rm c})\big]p_{\rm click}(k-1).
\end{align}
To solve this equation for any $k>2$, we need to calculate the starting point of the recursion, that is, $p_{\rm click}(2)$. 

The quantity $p_{\rm click}(2)$ consists of nine different cases, since each of the two conclusive results can correspond to a state $\ket{\varphi_j}$ with $j=0,1,2$. This is depicted in Fig.~\ref{fig:pclick2}.
\begin{figure}
\centering 
\centerline{\includegraphics*[scale=0.48]{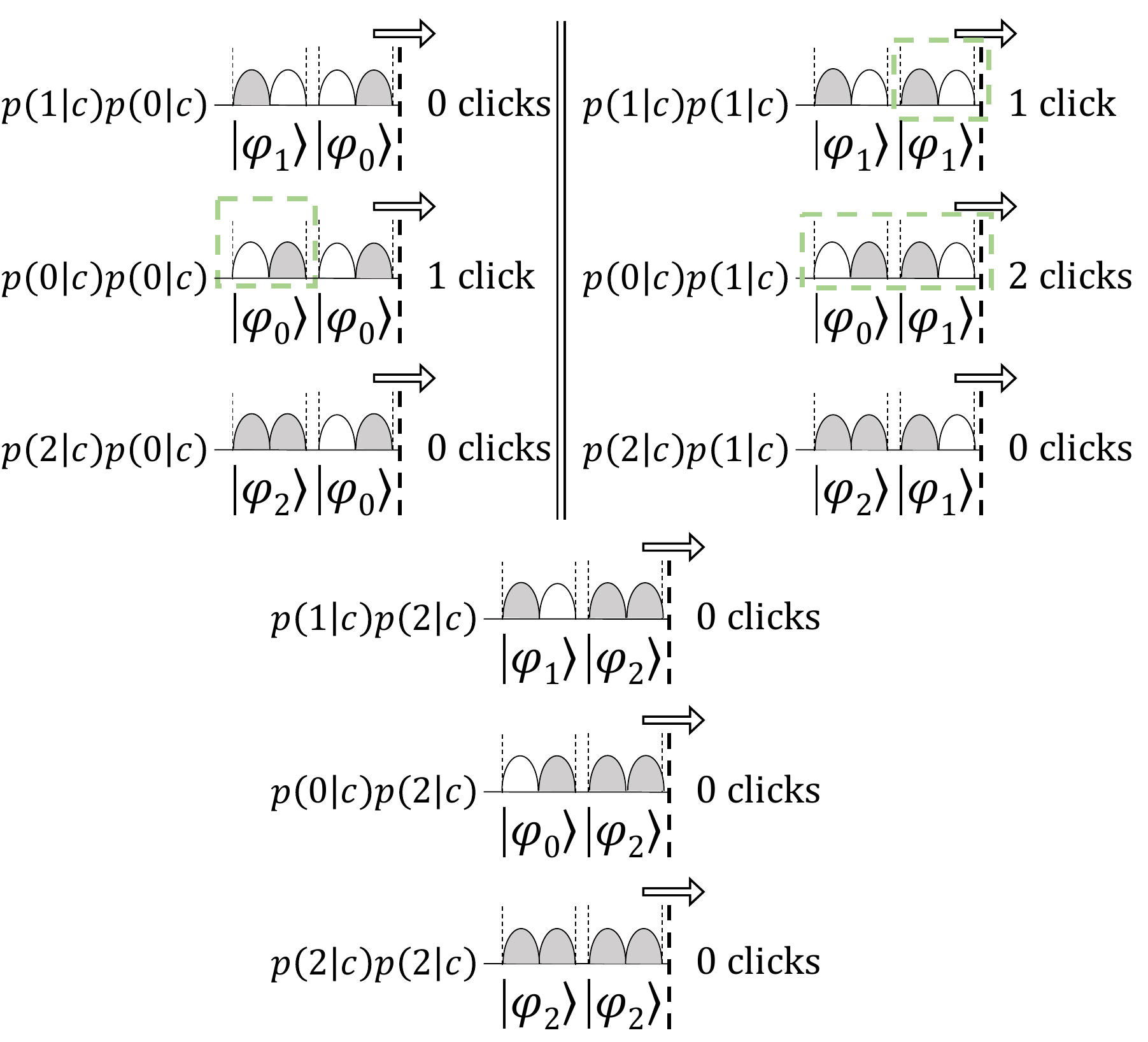}}
\caption{Illustration of the nine different cases corresponding to a block with $k=2$ consecutive conclusive measurement results, together with the number of ``clicks'' that Bob will obtain in his data line. For example, in the first case, with probability $p(1|{\rm c})p(0|{\rm c})$ the signals in the block are $\ket{\varphi_1}\ket{\varphi_0}$. This means that Eve cannot extract a sub-block surrounded by vacuum pulses. Therefore, she replaces these signals with two vacuum signals $\ket{\varphi_{\rm vac}}$, and the number of ``clicks'' at Bob's data line is then zero. The other cases are analogous. In the figure, the ovals represent the two optical pulses within a signal. }
\label{fig:pclick2}
\end{figure}
Whenever one out of the two signals in the block is $\ket{\varphi_2}$, Eve replaces all signals with vacuum signals $\ket{\varphi_{\rm vac}}$, as she cannot find a sub-block that has vacuum pulses on its borders. In this case, Bob will not observe any detection ``click''. For the same reason, if the first result of the block is $\ket{\varphi_0}$ and the second one is $\ket{\varphi_1}$, Eve also replaces these two signals with vacuum signals $\ket{\varphi_{\rm vac}}$. In all the other cases, Eve resends Bob one non-vacuum signal (and, thus, he obtains one detection ``click'') except when the first signal is $\ket{\varphi_1}$ and the second one is $\ket{\varphi_0}$, in which case she resends him these two signals (and, thus, Bob obtains two detection ``clicks''). We find, therefore, that 
\begin{equation}\label{pclick2}
p_{\rm click}(2)=4 p(1|{\rm c})^2,
\end{equation}
where we have used the fact that $p(0|{\rm c})=p(1|{\rm c})$. 

By combining Eq.~\eqref{final_pClick} with Eq.~\eqref{pclick2}, we finally obtain that $p_{\rm click}(k)$ satisfies
\begin{align}\label{FINAL_pClick}
p_{\rm click}(k)&=\frac{1}{p(1|{\rm c})}\Big\{-1+(k+1) p(1|{\rm c})\\\nonumber&+\big[1-2p(1|{\rm c})\big]^k\big[1+(k-1)p(1|{\rm c})\big]\Big\}.
\end{align}
for any $k\geq 2$.

\section{Performance evaluation}\label{Sec:Evaluation}

In this section, we evaluate the limitations that the zero-error attack presented in Sec.~\ref{Sec:Attack} imposes on both the maximum distance and secret key rate which might be achievable with COW-QKD. Moreover, we compare our results with other previously introduced zero-error attacks against this scheme. 

\subsection{Upper bound on the transmission distance}\label{Sec:Compare_with_previous}

Here, we first compare the value of the gain $G_{\rm zero}$ of the zero-error attack introduced in Sec.~\ref{Sec:Attack} with that associated to the zero-error attack in~\cite{cow_attack}, which has been shown to provide much tighter upper security bounds for COW-QKD than previous analyses~\cite{upp_cow2,upp_cow1,filteringattack}. 

For this, we consider, for instance, the most recent implementations of COW-QKD reported in~\cite{cow4}. The experimental parameters are provided in Table~\ref{tab:korzhParams}. They correspond to those experiments in~\cite{cow4} which use the highest and the lowest intensity value $\mu=|\alpha|^2$ for Alice's signals. 
\begin{table}
\centering
\begin{tabular}{ |c|c|c|c|c| } 
 \hline
 $\mu$ & Attenuation[dB] & Distance (km) &  $p_{\rm d}$ & $\eta_{\rm D}$   \\ 
 \hline
 0.06  & 16.9 & 104  & $4.38\times 10^{-7}$ & 0.22  \\ 
 \hline
 0.1   & 34.1 &  203 &  $1.3\times 10^{-8}$  & 0.27  \\ 
 \hline
\end{tabular}
\caption{Experimental parameters associated to those experiments reported in~\cite{cow4} that use the highest and the lowest intensity value $\mu=|\alpha|^2$ for Alice's signals. The attenuation included in the table only considers the channel loss; the distance corresponds to the fibre length; and $p_{\rm d}$ and $\eta_{\rm D}$ are, respectively, the dark count rate and the detection efficiency of Bob's detector in the data line. Moreover, in both experiments $f=0.155$, and the transmittance of Bob's beamsplitter is $t_{\rm B}=0.9$.}\label{tab:korzhParams}
\end{table}

The result of the comparison is shown in Table~\ref{tab:comparison}. This table demonstrates that the gain $G_{\rm zero}$ associated to the zero-error attack in this work can be more than an order of magnitude higher than that in~\cite{cow_attack}. 
\begin{table}
\centering
\begin{tabular}{ |l|c|c| } 
 \hline
       & $\log_{10}(G_{\rm zero})$ & $L_{\rm zero}$(km)  \\ 
 \hline
 Zero-error attack in~\cite{cow_attack} $\mu=0.06$ & -3.8 &   120  \\ 
 \hline
 This work $\mu=0.06$                  & -2.62  &  47   \\ 
 \hhline{|=|=|=|} 
Zero-error attack in~\cite{cow_attack} $\mu=0.1$ & -3.3 &    105  \\ 
 \hline
This work $\mu=0.1$                  & -2.19  &   38  \\ 
 \hline
\end{tabular}
\caption{Comparison between the zero-error attack presented in this work and that introduced in~\cite{cow_attack}. $L_{\rm zero}$ is the distance corresponding to $G_{\rm zero}$ given the experimental parameters in~Table~\ref{tab:korzhParams}. That is, no positive secret key rate is posible beyond $L_{\rm zero}$.}\label{tab:comparison}
\end{table}
This implies a significant reduction of the maximum achievable distance, which we shall call $L_{\rm zero}$, that is possible with COW-QKD in the absence of errors. To obtain $L_{\rm zero}$ from $G_{\rm zero}$ we use a typical channel model that is described in Appendix~\ref{app:sec_experiments_comparison}, and matches the specific experimental parameters given in the tables. For example, as illustrated in Table~\ref{tab:comparison}, if we consider the experimental implementation over 104 km (203 km) which was claimed to be secure in~\cite{cow4}, it turns out that Eve could perform a zero-error attack already at a distance $L_{\rm zero}=47$ km ($L_{\rm zero}=38$ km) according to this work. We note that the zero-error limit achieved with the attack in~\cite{cow_attack} was $L_{\rm zero}=120$ km ($L_{\rm zero}=105$ km). This is a remarkable improvement. 

The results above consider the use of ultra-low-loss optical fibres~\cite{cow_attack}, which might be challenging to employ in practical applications. To conclude this section, we now consider the maximum achievable distance that would be possible with COW-QKD by utilizing state-of-the-art devices but assuming the use of standard optical fibres with an attenuation coefficient of $0.2$ dB/km in the third telecom window. Also, for concreteness, we consider that the intensity of Alice signals is around $0.5$, which is similar to the value used for key generation in decoy-state QKD~\cite{decoy1,decoy2,decoy3}. The list of experimental parameters is provided in Table~\ref{tab:Example}.
\begin{table}
\centering
\begin{tabular}{ |c|c|c|c|c|c| } 
 \hline
 $\mu$ & $f$   & $\eta_{\rm D}$ &  $p_{\rm d}$  & $t_{\rm B}$ & $\alpha_{\rm att}$ (dB/km) \\ 
 \hline
 0.5   & 0.1~\cite{upp_cow2}   &    0.77~\cite{minder}        &      $2\times10^{-8}$~\cite{minder}         &     0.9~\cite{cow4}     & 0.2  \\  
 \hline
\end{tabular}
\caption{List of experimental parameters. Their meaning coincides with that given in Table~\ref{tab:korzhParams}. $\alpha_{\rm att}$ denotes the attenuation coefficient of the channel.}\label{tab:Example}
\end{table}
In this case, it turns out that $L_{\rm zero}$ is only about $22.60$~km. 

\subsection{Upper bound on the secret key rate}\label{Sec:UpperBound}

In this section, we evaluate the simple upper bound on the secret key rate $K$ of COW-QKD obtained in~\cite{cow_attack} by using the zero-error attack introduced above. The upper bound reads 
\begin{equation}\label{upp_bound_trivial}
K< (1-f)\eta \mu_{\rm max}(f)\equiv{}R_{\rm upp},
\end{equation}
where $\eta=\eta_{\rm ch}\eta_{\rm D}$ is the overall system's transmittance, with $\eta_{\rm ch}$ being the transmittance of the channel, and $\mu_{\rm max}(f)$ is the maximum allowed intensity for Alice's signals such that Eve's zero-error attack against all signals sent by Alice is not possible. That is, for each value $G$ of the gain at Bob's data line, $\mu_{\rm max}(f)$ is the maximum intensity of Alice's signals that guarantees $G_{\rm zero}<G$. Note that by increasing the signals' intensity, the success probability $p_{\rm c}$ of Eve's USD measurement increases as well, and, thus, also $G_{\rm zero}$ increases. We denote the maximum intensity by $\mu_{\rm max}(f)$ because it typically depends on $f$. Actually, Eq.~(\ref{upp_bound_trivial}) is just a simple upper bound on the probability that Alice sends Bob a signal that encodes a bit value ({\it i.e.}, $\ket{\varphi_0}$ or $\ket{\varphi_1}$) and Bob observes a ``click'' in his data line, which, obviously, is also an upper bound on the secret key rate. 

For the numerical simulations, we use the channel model described in Appendix~\ref{app:sec_experiments_comparison}. Moreover, we consider the best possible scenario for Alice and Bob for key generation, that is, we set the dark count probability of Bob's detectors to zero and assume that $t_{\rm B}\approx 1$, which means that almost all the incoming signals go into Bob's data line. If the intensity of Alice's signals is $\mu_{\rm max}(f)$, then the expected gain $G$ at Bob's data line has the form
\begin{equation}\label{gainfor_skr}
G=1-\left[(1-f)e^{-\eta\mu_{\rm max}(f)}+f e^{-2\eta\mu_{\rm max}(f)} \right].
\end{equation}
Then, for given $\eta$ and $f$, we determine numerically the maximum value of $\mu_{\rm max}(f)$ such that $G_{\rm zero}<G$. The result is shown in Fig.~\ref{fig:mumax}, which illustrates $\mu_{\rm max}(f)$ as a function of $\eta$ when $f=0.155$~\cite{cow4}. 
\begin{figure}
\centering 
\centerline{\includegraphics*[scale=0.42]{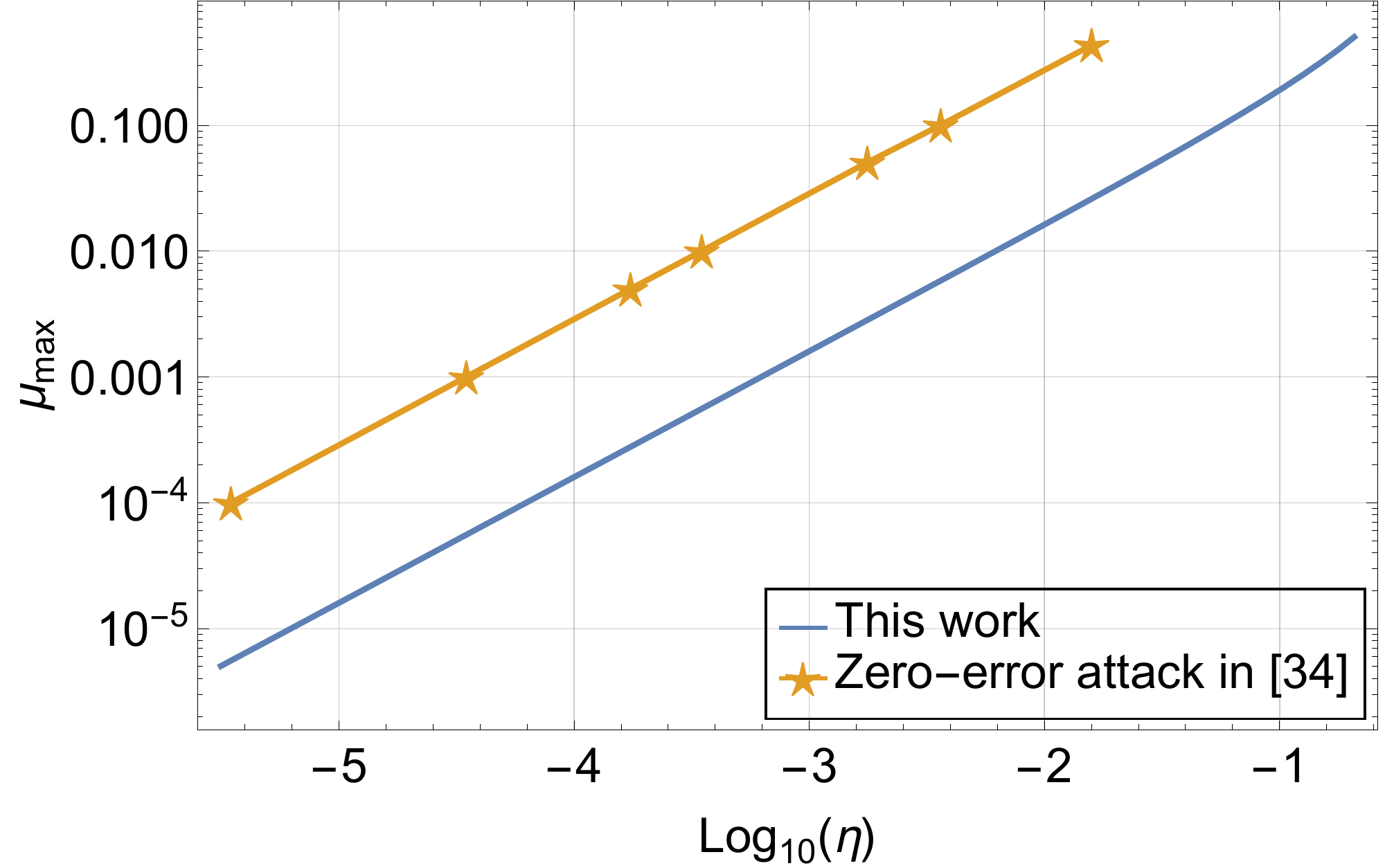}}
\caption{Maximum possible intensity of Alice's signals as a function of the overall system's transmittance $\eta$ when $f=0.155$~\cite{cow4}. For comparison, this figure includes as well the results in~\cite{cow_attack}; the stars correspond to the actual data points evaluated in that work and the line is an interpolation.}
\label{fig:mumax}
\end{figure}
For comparison, this figure also includes the results obtained in~\cite{cow_attack}. We can see that $\mu_{\rm max}(f)$ is very limited and decreases quite fast when $\eta$ decreases. Moreover, the maximum $\mu_{\rm max}(f)$ imposed by the zero-error attack presented in Sec.~\ref{Sec:Attack} is more than an order of magnitude lower than that in~\cite{cow_attack}. 

Given $f$ and $\mu_{\rm max}(f)$, we can evaluate the upper bound $R_{\rm upp}$ given by Eq.~(\ref{upp_bound_trivial}) as a function of $\eta$. This is illustrated in Fig.~\ref{fig:upperbound}.
\begin{figure}
\centering 
\centerline{\includegraphics*[scale=0.42]{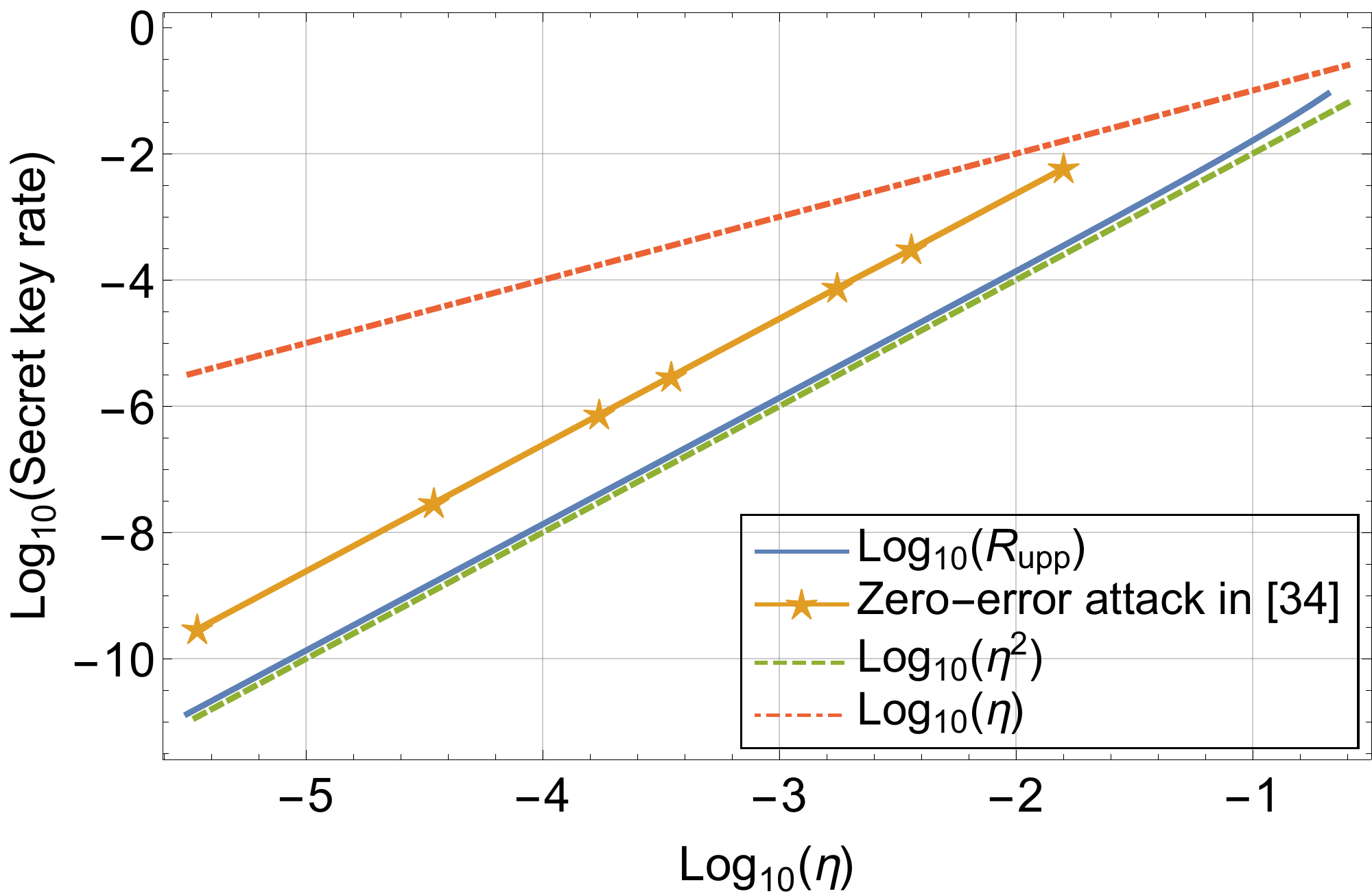}}
\caption{Upper bound $R_{\rm upp}$ on the secret key rate of COW-QKD as a function of $\eta$ when $f=0.155$~\cite{cow4}. For comparison, this figure includes as well the results in~\cite{cow_attack} (the stars correspond to the actual data points evaluated in that work and the line is an interpolation), together with the curves for linear and quadratic scaling in $\eta$.}
\label{fig:upperbound}
\end{figure}
The same improvement observed in $\mu_{\rm max}(f)$ when comparing the zero-error attack above and that in~\cite{cow_attack} is obviously also present in $R_{\rm upp}$. Indeed, this quantity now almost overlaps the curve $\log_{10}(\eta^2)$. Moreover, we note that by setting $f$ smaller than 0.155, $R_{\rm upp}$ moves slightly closer to the $\log_{10}(\eta^2)$ line, so decreasing $f$ would be basically unnoticeable in Fig.~\ref{fig:upperbound}.

\section{Discussion}\label{Sec:Discussion}

COW-QKD monitors eavesdropping through the error rates observed in the data line and in the monitoring line of Bob's receiver. However, according to zero-errors attacks, this is not sufficient to achieve a good performance. To improve the robustness of COW-QKD against this type of attacks and, thus, increase its achievable key rate and distance, Alice and Bob need to monitor more observables and include this additional information in the security proof. 

One possibility would be to modify Bob's receiver such that it can also measure the coherence between non-adjacent pulses, like it has been proposed in~\cite{dps4} for the case of DPS-QKD. Indeed, by doing so Alice and Bob could now avoid that Eve replaces Alice's original signals with vacuum signals $\ket{\varphi_{\rm vac}}$ without introducing errors, which is a key feature exploited by zero-error attacks. The main drawback of this approach is, however, that it requires a much more cumbersome receiver,  thus the principal advantage of COW-QKD regarding its simple experimental setup would probably vanish.  

A second option would be to measure the detection rates of Alice's signals at Bob's side. This has been also suggested in~\cite{cow_attack,upp_cow2}. For example, in the zero-error attack introduced in this work, Eve does not resend decoy signals to Bob. This is so because for the typical experimental parameter regime considered in COW-QKD ({\it i.e.}, when $f\leq 0.15$ and $|\alpha|^2\leq 0.5$), Eve's optimal USD measurement corresponds to the first case presented in the Claim in Sec.~\ref{meas}, where $q_{\rm s}^{\rm d}=0$. In this sense, the attack in Sec.~\ref{Sec:Attack} represents an extreme case where the detection statistics of the decoy signals are not preserved at all. 

Of course, if necessary, one could slightly modify the attack and impose that $q_{\rm s}^{\rm d}>0$ to guarantee that Eve resends some decoy signals to Bob. Given that $q_{\rm s}^{\rm d}$ is sufficiently small, it can be shown that the results obtained in such scenario would basically match those presented in this paper. The situation changes, however, if one requires that the detection rates associated to the different signals sent by Alice are similar to the expected values in the absence of Eve (see also~\cite{upp_cow2,filteringattack}). Indeed, this would have a big impact on the zero-error attack introduced above. First, the probability that Eve obtains a conclusive measurement result with her USD measurement would now decrease with respect to the optimal solution provided by the Claim. Also, Eve would have to change the post-processing of her measurement results to decide which signals she actually resends to Bob. Note that the current post-processing favours the transmission of bit signals with respect to decoy signals. This is so because bit signals contain a vacuum optical pulse, and thus it is easier for Eve to find sub-blocks of signals surrounded by vacuum pulses if they include bit signals. As a result, one expects that the gain $G_{\rm zero}$ at which a zero-error attack is actually possible would decrease significantly. 

A main difficulty of this second approach is, however, how to incorporate the detection rate information in a security proof for COW-QKD. For example, Eve might attack a small fraction of the signals sent by Alice, and thus detecting deviations between the actual detections rates and the expected ones might be challenging in practice. Moreover, we note that this problem could be amplified because of statistical fluctuations. Indeed, since in any practical implementation Alice sends a finite number of signals to Bob, the observed detection rates will naturally deviate from the expected ones even in the absence of Eve, and Eve could try to hide her attack in such deviations. 

\section{Conclusion}\label{Sec:Conclusion}

In this paper, we have proposed a simple, and essentially optimal, zero-error attack against coherent-one-way (COW) quantum key distribution (QKD). In this attack, Eve measures out all the signals sent by Alice one by one by employing an optimal unambiguous state discrimination measurement. Afterwards, she sends Bob all those blocks of signals which do not introduce any error in his data line nor in his monitoring line. Importantly, zero-error attacks are a special type of intercept-and-resend attack and, thus, they do not allow the distribution of a secure key. 

In doing so, we have obtained upper security bounds on the secret key rate of COW-QKD that are more than an order of magnitude lower than previously known upper bounds. Our attack highlights the fact that only monitoring errors in Bob's data and monitoring lines is not sufficient to achieve a good performance with this protocol. 

\section*{Acknowledgments}

The authors would like to thank the Galician Regional Government (consolidation of Research Units: AtlantTIC), the Spanish Ministry of Economy and Competitiveness (MINECO), the Fondo Europeo de Desarrollo Regional (FEDER) through Grant No.~TEC2017-88243-R, and the European Union's Horizon 2020 research and innovation programme under the Marie Sk\l{}odowska-Curie grant agreement No 675662 (project QCALL) for financial support.

\appendix
\section{Proof of the Claim}\label{proof}

To prove the Claim in Sec.~\ref{meas}, we follow the steps presented in~\cite{sugimoto}. As we will show below, due to the symmetry of the signals ({\it i.e.}, the fact that $\innerproduct{\varphi_0}{\varphi_2}=\innerproduct{\varphi_1}{\varphi_2}$) and the symmetry of their a priori probabilities ({\it i.e.}, $p_{\ket{\varphi_0}}=p_{\ket{\varphi_1}}$), it turns out that the solution that maximizes $p_{\rm c}$ actually fullfils $p_{0|0}=p_{1|1}$. 

Precisely, we have that the inner products between Alice's signals satisfy
\begin{align}\label{App_InnerProd}
\innerproduct{\varphi_0}{\varphi_1}&=e^{-|\alpha|^2},\nonumber\\
\innerproduct{\varphi_0}{\varphi_2}&=\innerproduct{\varphi_1}{\varphi_2}=e^{-|\alpha|^2/2}.
\end{align}
These inner products are all positive, thus the following quantities can be defined
\begin{align}\label{App_alfas}
\alpha_1=\sqrt{\frac{\innerproduct{\varphi_0}{\varphi_1}\innerproduct{\varphi_0}{\varphi_2}}{\innerproduct{\varphi_1}{\varphi_2}}}&=e^{-|\alpha|^2/2},\nonumber\\
\alpha_2=\sqrt{\frac{\innerproduct{\varphi_0}{\varphi_1}\innerproduct{\varphi_1}{\varphi_2}}{\innerproduct{\varphi_2}{\varphi_0}}}&=e^{-|\alpha|^2/2},\nonumber\\
\alpha_3=\sqrt{\frac{\innerproduct{\varphi_1}{\varphi_2}\innerproduct{\varphi_2}{\varphi_0}}{\innerproduct{\varphi_0}{\varphi_1}}}&=1.
\end{align}

In the simplest case, we have that the optimized values can be written as~\cite{sugimoto}
\begin{align}\label{sol1}
p_{0|0}&=1-\alpha_1^2=1-e^{-|\alpha|^2},\nonumber\\
p_{1|1}&=1-\alpha_2^2=1-e^{-|\alpha|^2},\nonumber\\
p_{2|2}&=1-\alpha_3^2=0,
\end{align}
where, as already mentioned, we have that $p_{0|0}=p_{1|1}$. However, this result is only valid if the following three conditions are fulfilled~\cite{sugimoto}
\begin{align}
e^{-|\alpha|^2/2}\sqrt{p_{\ket{\varphi_0}}}&\leq e^{-|\alpha|^2/2}\sqrt{p_{\ket{\varphi_1}}}+\sqrt{p_{\ket{\varphi_2}}},\label{cond1_a}\\
e^{-|\alpha|^2/2}\sqrt{p_{\ket{\varphi_1}}}&\leq e^{-|\alpha|^2/2}\sqrt{p_{\ket{\varphi_0}}}+\sqrt{p_{\ket{\varphi_2}}},\label{cond1_b}\\
\sqrt{p_{\ket{\varphi_2}}}&\leq e^{-|\alpha|^2/2}\sqrt{p_{\ket{\varphi_0}}}+e^{-|\alpha|^2/2}\sqrt{p_{\ket{\varphi_1}}}\label{cond1_c},
\end{align}
where we have already plugged in the values of the parameters $\alpha_i$, with $i=1,2,3$, given by Eq.~(\ref{App_alfas}). Since $f\in(0,1)$ and $p_{\ket{\varphi_0}}=p_{\ket{\varphi_1}}$ we find that Eqs.~\eqref{cond1_a}-\eqref{cond1_b} are automatically satisfied. Regarding Eq.~\eqref{cond1_c}, if we insert the values of the probabilities $p_{\ket{\varphi_j}}$ in this equation, we have that it can be rewritten as
\begin{equation}\label{cond1}
\sqrt{\gamma}\leq e^{-|\alpha|^2/2},
\end{equation}
with $\gamma=f/[2(1-f)]$. That is, if Eq.~\eqref{cond1} is fulfilled then the optimal values for $q^{\rm s}_{\rm s}$ and $q^{\rm d}_{\rm s}$ are
\begin{align}\label{sol1_Final}
q^{\rm s}_{\rm s}&=1-e^{-|\alpha|^2},\ {\rm and}\ q^{\rm d}_{\rm s}=0.
\end{align}

On the other hand, if Eq.~\eqref{cond1} is not satisfied, it can be shown that the optimal values are the following~\cite{sugimoto} 
\begin{align}\label{sol2}
p_{0|0}&=1-\frac{\alpha_1}{\sqrt{p_{\ket{\varphi_0}}}}\left(-\alpha_2\sqrt{p_{\ket{\varphi_1}}}+\alpha_3\sqrt{p_{\ket{\varphi_2}}}\right),\nonumber\\
p_{1|1}&=1-\frac{\alpha_2}{\sqrt{p_{\ket{\varphi_1}}}}\left(-\alpha_1\sqrt{p_{\ket{\varphi_0}}}+\alpha_3\sqrt{p_{\ket{\varphi_2}}}\right),\nonumber\\
p_{2|2}&=1-\frac{\alpha_3}{\sqrt{p_{\ket{\varphi_2}}}}\left(\alpha_1\sqrt{p_{\ket{\varphi_0}}}+\alpha_2\sqrt{p_{\ket{\varphi_1}}}\right).
\end{align}
Since $\alpha_1=\alpha_2$ and $p_{\ket{\varphi_0}}=p_{\ket{\varphi_1}}$, we have also here that $p_{0|0}=p_{1|1}$. By inserting in Eq.~\eqref{sol2} the values of the coefficients $\alpha_i$ given by Eq.~\eqref{App_alfas} and those of the probabilities $p_{\ket{\varphi_j}}$, we obtain
\begin{align}\label{sol2_Final}
q^{\rm s}_{\rm s}&=1+e^{-|\alpha|^2}-e^{-|\alpha|^2/2}\sqrt{\frac{2f}{1-f}},\nonumber\\
q^{\rm d}_{\rm s}&=1-\sqrt{\gamma^{-1}}e^{-|\alpha|^2/2}.
\end{align}
We note, however, that for this solution to be valid, we need that $q^{\rm s}_{\rm s}$ and $q^{\rm d}_{\rm s}$ are non-negative. Since Eq.~\eqref{cond1} is not satisfied, {\it i.e}, $\sqrt{\gamma}> e^{-|\alpha|^2/2}$, it follows that $q^{\rm d}_{\rm s}>0$. Likewise, it is easy to show that $q^{\rm s}_{\rm s}\geq 0$ is equivalent to the following condition
\begin{equation}\label{cond2a}
\cosh\left(\frac{|\alpha|^2}{2}\right)\geq\sqrt{\gamma}.
\end{equation} 
That is, we find that Eq.~\eqref{sol2_Final} provides the optimal values for $q^{\rm s}_{\rm s}$ and $q^{\rm d}_{\rm s}$ if Eq.~\eqref{cond1} is not satisfied and Eq.~\eqref{cond2a} holds. 

Finally, let us consider the case where Eqs.~\eqref{cond1} and \eqref{cond2a} do not hold. This scenario can be solved by using a two step procedure~\cite{sugimoto}. First, one sets, for example, $p_{0|0}=0$ and reduces the problem to a two-state USD problem. And, second, one infers the solution to the original three-state USD problem based on the solution to such two-state USD problem.

Precisely, we have that the two states and their a priori probabilities in the reduced problem are given by~\cite{sugimoto}
\begin{align}\label{reduced_2states}
\ket{\varphi_1'}&=\frac{\ket{\varphi_0}-\ket{\varphi_1}e^{-|\alpha|^2}}{\sqrt{1-e^{-2|\alpha|^2}}},\nonumber\\ 
\ket{\varphi_2'}&=\frac{\ket{\varphi_2}-\ket{\varphi_1}e^{-|\alpha|^2/2}}{\sqrt{1-e^{-|\alpha|^2}}},
\end{align}
and
\begin{align}\label{reduced_apriori}
p_{\ket{\varphi_1'}}&=\frac{\frac{1-f}{2}\left(1-e^{-2|\alpha|^2}\right)}{\frac{1-f}{2}\left(1-e^{-2|\alpha|^2}\right)+f(1-e^{-|\alpha|^2})},\nonumber\\ 
p_{\ket{\varphi_2'}}&=\frac{f(1-e^{-|\alpha|^2})}{\frac{1-f}{2}\left(1-e^{-2|\alpha|^2}\right)+f(1-e^{-|\alpha|^2})},
\end{align}
respectively. 

In this reduced problem, we define $p_{1|1}'$ and $p_{2|2}'$ as the probabilities that Eve correctly identifies the states $\ket{\varphi_1'}$ and $\ket{\varphi_2'}$, respectively. This means that our goal is to maximize the probability $p_{1|1}'p_{\ket{\varphi_1'}}+p_{2|2}'p_{\ket{\varphi_2'}}$ to obtain a conclusive result. Since Eq.~\eqref{cond2a} is not satisfied, it can be shown that the following condition holds
\begin{equation}\label{solutionReduced2Cond}
\sqrt{\frac{p_{\ket{\varphi_2'}}}{p_{\ket{\varphi_1'}}}}>\frac{1}{\innerproduct{\varphi_1'}{\varphi_2'}}.
\end{equation}
In this situation, the optimal probabilities for successfully identifying the states $\ket{\varphi_1'}$ and $\ket{\varphi_2'}$ are given by~\cite{sugimoto}
\begin{eqnarray}\label{sol_Reduced}
p_{1|1}'&=&0,\nonumber\\
p_{2|2}'&=&1-|\innerproduct{\varphi_1'}{\varphi_2'}|^2=1-\frac{e^{-|\alpha|^2}\left(1-e^{-|\alpha|^2}\right)}{1-e^{-2|\alpha|^2}}.\ \ \ \ \ \
\end{eqnarray}
This means, in particular, that the optimal probabilities for the original three-state USD problem can be obtained as follows~\cite{sugimoto}
\begin{align}\label{optOrig}
p_{1|1}&=p_{1|1}'\expval{Q}{\varphi_1}=0,\nonumber\\
p_{2|2}&=p_{2|2}'\expval{Q}{\varphi_2},
\end{align}
where $Q=\openone-\ketbra{\varphi_1}$. Since $\expval{Q}{\varphi_2}=1-e^{-|\alpha|^2}$, we find that the optimal solution to the original problem is
\begin{eqnarray}\label{sol3_final}
q^{\rm s}_{\rm s}&=&0,\nonumber\\
q^{\rm s}_{\rm d}&=&\left(1-e^{-|\alpha|^2}\right)\left[1-\frac{e^{-|\alpha|^2}\left(1-e^{-|\alpha|^2}\right)}{1-e^{-2|\alpha|^2}}\right]\nonumber\\
&=&\tanh\left(\frac{|\alpha|^2}{2}\right), 
\end{eqnarray}
which again satisfies $p_{0|0}=p_{1|1}$. 

The uniqueness of the solution above is also proven in~\cite{sugimoto}. This concludes the proof of the Claim in Sec.~\ref{meas}. 

\section{Channel model}\label{app:sec_experiments_comparison}

For simplicity, we consider a lossy channel with transmittance $\eta_{\rm ch}=10^{-\alpha_{\rm att} L/10}$, where $\alpha_{\rm att}$ denotes its attenuation coefficient measured in dB/km and $L$ corresponds to the transmission distance measured in km. Moreover, we disregard any misalignment effect. This means that the expected gain at Bob's data line can be expressed as
\begin{equation}\label{gainExp_SM}
G=1-(1-p_{\rm d})\big[(1-f)e^{-\mu t_{\rm B}\eta_{\rm D}\eta_{\rm ch} }+f e^{-2\mu t_{\rm B} \eta_{\rm D} \eta_{\rm ch}} \big],
\end{equation} 
where $p_{\rm d}$ and $\eta_{\rm D}$ denote, respectively, the dark count rate and the detection efficiency of Bob's detector in the data line, $f$ is the probability that Alice emits a decoy signal $\ket{\varphi_2}$, $\mu=|\alpha|^2$ is the intensity of Alice's emitted coherent states, and $t_{\rm B}$ denotes the transmittance of Bob's beamsplitter.

Eq.~\eqref{gainExp_SM} can be interpreted as calculating the probability of not having a ``click'' at all in Bob's data line and subtracting this quantity from probability one. Note that the term $1-p_{\rm d}$ is the probability that there is no ``click'' in Bob's detector due to dark counts, $(1-f)e^{-\mu t_{\rm B}\eta_{\rm D}\eta_{\rm ch} }$ is the joint probability that Alice emits a signal state $\ket{\varphi_0}$ or $\ket{\varphi_1}$ and this signal does not produce a ``click'' at Bob's data line, and $f e^{-2\mu t_{\rm B} \eta_{\rm D} \eta_{\rm ch}}$ is the joint probability that Alice emits a decoy signal $\ket{\varphi_2}$ and there is no ``click'' at Bob's data line either.
 
To obtain $L_{\rm zero}$ from $G_{\rm zero}$ in Sec.~\ref{Sec:Evaluation}, one simply substitutes in Eq.~\eqref{gainExp_SM} $G$ with $G_{\rm zero}$ and then obtains the associated transmission distance $L$ (which now corresponds to $L_{\rm zero}$), given the experimental parameters $p_{\rm d}$, $\eta_{\rm D}$, $f$, $\mu$, $t_{\rm B}$ and $\alpha_{\rm att}$.

\newpage 

\end{document}